\documentclass{amsart}
\usepackage{amsmath,amsthm,amssymb,amsfonts,graphicx,mathrsfs,bm,amscd,latexsym}
\usepackage[linesnumbered, boxed,lined]{algorithm2e}
\RestyleAlgo{ruled}
\usepackage{color,enumitem,hyperref}
\usepackage{blindtext}
\usepackage{textcomp}
\usepackage[all]{xy}

\newtheorem{thm}{Theorem}[section]

\newtheorem{lem}[thm]{Lemma}

\theoremstyle{definition}
\newtheorem{rem}[thm]{Remark}

\def\R{\mathbb{R}}

\def\N{\mathbb{N}}

\def\e{\epsilon}
\def\g{\gamma}

\def\a{\alpha}

\DeclareMathOperator{\card}{card}
\DeclareMathOperator{\dist}{dist}

\newcommand{\Flat}{\mathsf{Flat}}

\numberwithin{equation}{section}

\begin{document}

\title[Time complexity of the ATSP algorithm]{Time complexity of the Analyst's Traveling Salesman algorithm}

\author{Anthony Ramirez}
\email{aramir13@utk.edu}
\address{Department of Mathematics, The University of Tennessee, Knoxville, TN  37966}

\author{Vyron Vellis}
\email{vvellis@utk.edu}
\address{Department of Mathematics, The University of Tennessee, Knoxville, TN  37966}

\thanks{V. V. was partially supported by the NSF DMS grant 1952510.}


\date{\today}


\begin{abstract}
The Analyst's Traveling Salesman Problem asks for conditions under which a (finite or infinite) subset of $\R^N$ is contained on a curve of finite length. We show that for finite sets, the algorithm constructed in \cite{Schul-Hilbert,BNV} that solves the Analyst's Traveling Salesman Problem has polynomial time complexity and we determine the sharp exponent.
\end{abstract}

\maketitle

\section{Introduction}

The \emph{(Euclidean) Traveling Salesman Problem (TSP)} \cite{lawler,gutin,applegate} asks to find the shortest path through a set $V$ of $n$ points in $\R^N$ that starts and ends at a given vertex $v_0$ of $V$.  Apart from its natural applications in itinerary design, and its influence on operations research, and polyhedral theory, in the last 50 years the TSP has gained great prominence in computer science due to its immense computational complexity. For example, it is well known that the TSP is NP-hard \cite{Gar76, Pap77}; that is, it is at least as difficult as the hardest problems in NP (the class of all problems that can be solved by a non-deterministic Turing machine in polynomial time $O(n^k)$ for some $k\in\N$).

What is the minimum amount of time (depending on $n$) which is required to obtain a solution of the TSP? Can it be polynomial (i.e. $O(n^k)$ for some $k\in\N$)? Obviously, one could simply try all possible paths but that would require time comparable to $(n-1)!$ which is far from being polynomial. The Bellman-Held-Karp algorithm \cite{Bellman,HK} improves the latter bound to $O(n^2 2^n)$. However, it is still unknown whether a polynomially complex algorithm exists. 

In lieu of the above discussion, ``nearly-optimal" algorithms have been explored. That is, algorithms that produce a path which may not be the optimal one but is comparable (or even arbitrarily close) in length to the optimal one. The \emph{nearest insertion algorithm} \cite{RSL2} computes in $O(n^2)$ a path which is at most twice in length of the optimal one; see also \cite{GBDS} and \cite{LK} for heuristics of similar time complexity. Christofides heuristic \cite{Christ} computes in $O(n^3)$ time a path that is at most $3/2$ times the length of the optimal one. Grigni, Koutsoupias, and Papadimitriou \cite{GKP} designed an algorithm for planar graphs that for all $\e>0$ provides in $O(n^{O(1/\e)})$ time a tour of length at most $(1+\e)$ times the optimal one. A similar result of $O(n^{O(1/\e)})$ time was also obtained later by Mitchel \cite{M}. In his celebrated work, Arora \cite{Arora} constructed an algorithm which, for each $\e>0$, gives with probability more than $1/2$ a path which is at most $(1+\e)$ times the length of the optimal one in $O(n(\log{n})^{(O(\sqrt{N}/\e))^{N-1}})$ time. Arora's algorithm can be derandomized by increasing the running time by O$(n^N)$.  See also \cite{BGK} for relevant results.

The \emph{Analyst's Traveling Salesman Problem (ATSP)} \cite{Jones-TST} is a generalization of the TSP where it is asked to find a curve of finite length that contains a given (finite or infinite) set $V \subset \R^N$. While the TSP (which is the finite case of ATSP) always admits a solution, the same is not true in general for the ATSP. For instance, if $V$ is an unbounded set, then clearly every curve that contains $V$ must be unbounded itself, hence with infinite length. Perhaps less intuitively, smaller, but still infinite, sets may not admit a solution. For example, if $V$ is the set of all points in the unit square $[0,1]^2$ with rational coordinates, it is not hard to see that $V$ is bounded and countable but there exists no rectifiable curve that contains $V$. 

The classification of sets in $\R^N$ for which the ATSP can be solved, was done by Jones \cite{Jones-TST} (for $\mathbb{R}^2$) and by Okikiolu \cite{Ok-TST} (for higher dimensional spaces). Ever since, this classification has played a crucial role in the development of geometric measure theory, and has found applications in many facets of analysis including complex analysis, dynamics, harmonic analysis and metric geometry. The work of Jones and Okikiolu provides an algorithm that, for any finite set $V \subset \R^N$ yields a tour that is at most $C(N)$ times longer than the optimal path. Here $C(N)$ is a constant that depends only on the dimension $N$. The Jones-Okikiolu algorithm is based on a local version of the Farthest Insertion algorithm \cite{JM}.

Later, Schul \cite{Schul-Hilbert} provided a modification of the algorithm so that the ratio of the length of the yielded path over the length of the optimal path is bounded by a constant $C$ independent of the dimension $N$. Variation of this algorithm also appears in \cite{BNV}. Here and for the rest of this paper, we refer to any of these two variations as the ATSP algorithm. The purpose of this note is to show that the ATSP algorithm, in the case that $V$ is finite, has polynomial time complexity.

\begin{thm}\label{thm:main}
The time complexity of the ATSP algorithm is $O(n^3)$.
\end{thm}

We remark that although time-wise this algorithm is fast, the ratio constant of the yielded path over the optimal path has not been computed and is much larger that Christofides' $3/2$ ratio. In fact, in our algorithm, the yielded path has length at 
%least 2 times and at 
most $(300)^{9/2}\log{300}$ times the length of the smallest spanning tree. Moreover, the exponent $3$ in our theorem is sharp and can not be lowered; see Section \ref{sec:sharp}.

An interesting connection to the Jones-Scul algorithm was given by Gu, Lutz, and Mayordomo \cite{glm06} who classified the sets $V$ that admit a solution to a computable extension of the ASTP. This variant of the problem characterizes the sets which are contained in a rectifiable computable curve. As we are concerned only with finite sets here, our algorithm already produces computable curves.

\subsection{Outline of the ATSP algorithm}\label{sec:outline}
Fix a set $V = \{v_1,\dots,v_n\}\subset \mathbb{R}^N$. The algorithm (described in detail in Section \ref{sec:graphs}) computes connected graphs $G_k=(V_k,E_k)$ with $k=1,\dots,m$ for some $m\leq n$, so that the last graph satisfies $V_m=V$. First, we compute $R_0 = 5\max\{|v| : v\in V\}$ which takes $O(n)$ time. It is clear that $V \subset [-R_0,R_0]^N$. The construction of graphs now is roughly as follows.

\subsubsection*{Step 1:} For the first graph we have $V_1=\{v_1\}$ and $E_1 =\emptyset$. We set $n_1=1$ and $U_1=V\setminus V_1$. If $U_1 = \emptyset$, then proceed to the \emph{Final Step}. Otherwise proceed to the next step.

\subsubsection*{Step $k+1$:} Inductively, we assume that for some $k\in\N$, we have defined an integer $n_k\in\N$, two sets $V_k, U_k \subset V$, and a graph $G_{k} = (V_k,E_k)$ such that $U_k = V \setminus V_k$, $U_k \neq \emptyset$, and $V_k$ is a maximal $(2^{-n_k}R_0)$-separated set. We use results in Section \ref{sec:nets} to define an integer $n_{k+1} > n_k$, a set $V_{k+1}$ that contains $V_k$ and is contained in $V$, and a set $U_{k+1} = V\setminus V_{k+1}$ such that $V_{k+1}$ is a maximal $(2^{-n_{k+1}}R_0)$-separated set and $\dist_H(V_k,V_{k+1}) <  2^{-n_{k+1}}R_0$. Here and for the rest of this paper, $\dist_H(A,B)$ denotes the Hausdorff distance between closed sets $A,B\subset \R^N$
\[ \dist_H(A,B) = \max\left\{ \sup_{x\in A}\inf_{y\in B}|x-y|, \sup_{y\in B}\inf_{x\in A}|x-y|\right\}.\]

Next, using results in Section \ref{sec:flat}, we define for each $v\in V_k$ a number $\a$ that measures how ``flat'' the set $V_{k+1}$ is around the point $v$. This notion of flatness is inspired by the \emph{Jones beta numbers} in the work of Jones \cite{Jones-TST} and Schul \cite{Schul-Hilbert}. 

Next, following the constructions in \cite{Schul-Hilbert} and \cite{BNV}, we create a connected graph $G_{k+1} = (V_{k+1},E_{k+1})$. Roughly speaking, if around some node $v$ of $G_{k}$, the set $V_{k+1}$ looks flat (i.e. the number $\a$ is small), then the graph $G_{k+1}$ around $v$ should also be flat; otherwise, we add new edges around $v$ so that the new graph is connected. We take care so that the total number of new edges is at most twice the total number of new points. 

\subsubsection*{Final Step:} Arriving at this step, we have created a graph $G_m = (V_m, E_m)$ with $V_m=V$. Since each $V_k$ is different than $V_{k-1}$, we arrive at the final step in at most $n$ steps. We show that each step above takes $O(n^2)$ time, and therefore, we obtain $G_m$ in $O(n^3)$ time. Now we use an algorithm in Section \ref{sec:tour} to parameterize $G_m$ in $O(n^3)$ time.

\subsection{H\"older curves} It should be noted that the work of Jones \cite{Jones-TST}, Okikiolu \cite{Ok-TST}, and Schul \cite{Schul-Hilbert} provides the sets for which a solution exists but not the solution itself when the set is infinite. That is, they classify the sets which are contained in curves of finite length, but do not provide the parametrization of the curves. Recently, the second named author with Badger and Naples \cite{BNV}, constructed an algorithm that, for those sets $V$ that have a solution in ATSP, produces a solution path $f:[0,1]\to \R^N$ which is at most $C$ times longer than the optimal path with $C$ independent of $N$. In this algorithm, one obtains parameterizations $(f_k)$ for \emph{all graphs} $G_k$ and not just one of them. Moreover, the parameterizations obtained satisfy the inequality $\|f_k-f_{k+1}\|_{\infty} \leq 2^{-n_{k+1}}$ for all $k$. Although this approach is more complicated than the one of Jones and Okikiolu, it comes very handy in designing H\"older parameterizations of sets such as the space-filling Peano curve. For finite sets, the algorithm of constructing $(f_k)$ has also polynomial complexity (and in fact with the same exponent) but we will not pursue this direction here.

\section{Preliminaries on graphs}\label{appendix2}

A \emph{(combinatorial) graph} is a pair $G=(V,E)$ of a finite vertex set $V$ and an edge set 
\[E \subset \{\{v,v'\} : v,v' \in V\text{ and }v\neq v'\}.\]
A graph $G' = (V',E')$ is a \emph{subgraph} of $G=(V,E)$ (and we write $G\subset G'$) if $V'\subset V$ and $E'\subset E$. A \emph{path} in $G$ is a set $ \g = \{\{v_1,v_2\}, \dots, \{v_{n-1},v_n\}\} \subset E$; in this case we say that $\g$ joins $v_1$, $v_n$. A combinatorial graph $G = (V,E)$ is connected, if for any distinct $v,v' \in V$ there exists a path $\g$ in $G$ that joins $v$ with $v'$. A \emph{connected component} of a combinatorial graph $G$ is a maximal connected subgraph of $G$.

\subsection{Components of a graph} Given a graph $G=(V,E)$ we describe an algorithm that returns the connected components of $G$. 

\begin{lem}\label{lem:graph1}
There exists an algorithm that determines the components of $G$ in $O((\card{V})(\card{E})^2)$ time.
\end{lem}

\begin{proof}
We fix $v_0 \in V$ and let $(V_{1,1} , E_{1,1}) = (\{v_0\},\emptyset)$. Assuming that for some $i \leq \card{E}-1$ we have defined $(V_{1,i} , E_{1,i})$, there are two possible cases.
\begin{enumerate}
\item If there exists $e \in E \setminus E_{1,i}$ such that $e$ has an endpoint in $V_{1,i}$, then we set $V_{1,i+1} = V_{1,i}\cup e$ and $E_{1,i+1} = E_{1,i}\cup \{e\}$.
\item If there does not exist $e \in E \setminus E_{1,i}$ such that $e$ has an endpoint in $V_{1,i}$, then we stop the procedure and set $V_{1} = V_{1,i}$ and $E_{1} = E_{1,i}$.
\end{enumerate}
For each $i$, we make sure in $O(\card{V_1})$ that no vertex appears twice. The time needed for the calculation of $V_1$ is at most a constant multiple of
\[ \sum_{i=1}^{\card{E}_1} (\card{V_{1,i}})\card{E}  = O((\card{V_1})(\card{E})^2).\]

If $V_1 \neq V$, we replace $(V,E)$ by $(V\setminus V_1, E\setminus E_{1})$ and we repeat the process again to obtain $(V_2,E_2)$. We continue this way until we exhaust all points of $V$. Since the $V_i$'s are disjoint, the sum of their cardinalities is $\card V$. Thus, the total time needed for the algorithm to complete is at most a constant multiple of 
\[ O((\card{V_1})(\card{E})^2) + O((\card{V_2})(\card{E})^2) + \cdots = O((\card{V})(\card{E})^2). \qedhere\]
\end{proof}

\subsection{Two-to-one tours on connected graphs}\label{sec:tour}
Given a connected graph $G=(V,E)$, we describe here an algorithm that gives a tour over all edges of $G$ so that each edge is traversed exactly twice and once in each direction.

\begin{lem}\label{lem:tour}
Let $G=(V,E)$ be a finite connected graph and let $v_0\in V$. There exists an algorithm which in $ O((\card{V})(\card{E})^2)$ time produces a finite sequence $(a_i)_{i=1}^{2M+1}$ of points in $V$ with $M=\card{E}$ that satisfies the following properties.
\begin{enumerate}
\item We have $\{a_1,\dots,a_{2M+1}\} = V$ with $a_1=a_{2M+1} = v_0$.
\item For any $i\in\{1,\dots,2M\}$, $\{a_i,a_{i+1}\} \in E$. 
\item For any $e \in E$ there exist exactly two distinct $i,j \in \{1,\dots,2M\}$ such that
\[ \{a_i,a_{i+1}\} = \{a_j,a_{j+1}\} = e.\]
Moreover, $a_i=a_{j+1}$ and $a_j = a_{i+1}$.
\end{enumerate}
\end{lem}

\begin{proof}
Set $a_{0,1}=v_0$ and $E_{0}= E$. 

Assume that for some odd $k \in \{0,1,\dots, M-1\}$ we have defined a set $E_k\subset E$ and a finite sequence $(a_{k,i})_{i=1}^{2k+1}$ with the following properties.
\begin{enumerate}
\item The set $\{a_{k,1},\dots,a_{k,2k+1}\} \subset V$ with $a_{k,1}=a_{k,2k+1} = v_0$.
\item For any $i\in\{1,\dots,2k\}$, $\{a_{k,i},a_{k,i+1}\} \in E\setminus E_k$. 
\item For any $e \in E \setminus E_k$ there exist exactly two distinct $i,j \in \{1,\dots,n-1\}$ such that
\[ \{a_{k,i},a_{k,i+1}\} = \{a_{k,j},a_{k,j+1}\} = e.\]
Moreover, $a_i=a_{j+1}$ and $a_j = a_{i+1}$.
\end{enumerate}

Find $e \in E_k$ and $j\in \{0,1,\dots, 2k+1\}$ such that $a_{k,j}\in e$. Since $G$ is assumed connected, such $j$ and $e$ exist and the search for them would take at most $O(\card{E}\card{V})$ time. Let $v$ be the unique element of $e\setminus \{a_{k,j}\}$. Define now $E_{k+1} = E_k \setminus \{e\}$ and 
\[ a_{k+1,i} =
\begin{cases}
a_{k,i} &\text{ if $1\leq k \leq j$}\\
v &\text{ if $i=j$}\\
a_{k,i-2} &\text{ if $j+2 \leq i \leq 2k+3$}
\end{cases}.\]
The construction of $E_{k+1}$ and the sequence $(a_{k+1,i})_{i=1}^{2k+3}$ takes $O(\card{E})$ time. By design and the inductive hypothesis, each edge in $E\setminus E_{k+1}$ is traversed exactly twice. In particular, the properties (1)--(3) above are true for $k+1$.

This process terminates when $k=M$. The total time required is a constant multiple of
\[ \card{E} \left( O(\card{E}) + O(\card{E}\card{V}) \right) = O(\card{V}(\card{E})^2).\]
For $i=1,\dots,2M+1$ set $a_i = a_{M,i}$. Note that $E_{M} = \emptyset$. By induction, every edge of $E = E\setminus E_{M}$ is traversed twice, once in each direction. Since we traverse every edge of $E$, we have that $\{a_1,\dots,a_{2M+1}\} = V$ and the proof is complete.
\end{proof}

\section{Construction of scales and nets}\label{sec:nets}

Given $\e>0$, we say that a set $X \subset \R^N$ is \emph{$\e$-separated} if for any $x,y \in X$ we have $|x-y|\geq \e$. Given a set $V\subset \R^N$ and $\e>0$, we say that a set $X\subset V$ is an \emph{$\e$-net} if it is a maximal $\e$-separated subset of $V$. 

Recall the definition of Hausdorff distance from Section \ref{sec:outline}. In the next lemma we describe an algorithm that gives nets for a given set $V$.

\begin{lem}\label{lem:nets}
Let $V\subset \R^N$ be a set of $n$ elements, let $\e>0$, let $X$ be an $\e$-net of $V$, and let $U = V \setminus X$. Assume that $U\neq \emptyset$. Then, in $O(n^2)$ time we can compute an integer $k\in\N$, a $(2^{-k}\e)$-net $X'$ of $V$, and a set $U' = V \setminus X'$ such that $X\subsetneq X' \subset V$ and
\begin{equation}\label{eq2}
2^{-k}\e \leq \dist_H(X,X') <  2^{1-k}\e.
\end{equation}
\end{lem}

\begin{proof}
Assume that $X=\{x_{1},\dots,x_{l}\}$ and $U = \{u_{1},\dots,u_{n-l}\}$, with $l\in \{1,\dots,n\}$, such that $U = V\setminus X$. We first calculate
\[ d = \max_{i=1,\dots, n-l} \min_{j=1,\dots,l}|x_{j}-u_{i}|.\]
Since both sets $X,U$ have cardinalities at most $n$, the computation of $d$ takes at most $O(n^2)$ time. Let $k$ be the smallest integer such that $2^{-k}\e \leq d$; the computation of $k$ clearly takes $O(1)$ time. By minimality of $k$ we have that $2^{1-k}\e > d$.

We initiate the construction by setting $a_{i}=x_{i}$ for $i\in\{1,\dots,l\}$,  $c_{j}=u_{j}$ for $j\in\{1,\dots,n-l\}$. We also set $X'_{1} = \{a_{1},\dots,a_{l}\}$, $U_1' = \emptyset$, and $U_{1}'' = \{c_{1},\dots,c_{n-l}\}$. This requires $O(n)$ time.

Inductively, assume that after $m$ steps, with $m\in\{1,\dots,n-l\}$, we have defined three disjoint sets
\[ X'_{m} = \{a_{1}, \dots,a_{p}\}, \quad U_{m}'=\{b_{1},\dots,b_{l+m-p-1}\}, \quad U_{m}''=\{c_{1},\dots,c_{n-l-m+1}\}\]
such that $(X_{m}' \setminus X) \cup U_{m}' \cup U_{m}''= U$. Conventionally, if $l+m-p-1=0$, then we assume $U_{m}' = \emptyset$. If 
\begin{equation}\label{eq1}
\min_{i=1,\dots,p}|c_{n-l-m+1}-a_{i}| \geq 2^{-k}\e,
\end{equation}
then we set $a_{p+1}=c_{n-l-m+1}$ and
\begin{align*}
X_{m+1}' = X_{m}'\cup\{a_{p+1}\}, \quad U_{m+1}' = U_{m}', \quad U_{m+1}'' = U_m'' \setminus \{c_{n-l-m+1}\}.
\end{align*}
Otherwise, we set $b_{l+m-p}=c_{n-l-m+1}$ and
\begin{align*}
X_{m+1}' = X_{m}', \quad U_{m+1}' = U_{m}'\cup\{b_{l+m-p}\}, \quad U_{m+1}'' = U_m'' \setminus \{c_{n-l-m+1}\}.
\end{align*}
Note that the calculation of (\ref{eq1}) as well as the definition of sets $X_{m+1}', U_{m+1}', U_{m+1}''$ requires $O(n)$ time to complete.

Set $X' = X_{n-l+1}$ and $U'=U_{n-l+1}'$. The construction of $X',U'$ takes $O(n^2)$ time to complete. It remains to verify the conclusions of the lemma.

By the definition of $d$, there exists at least one element $u$ in $U$ such that $|x-u| \geq 2^{-k}\e$ for all $x\in X$. Therefore, the set $X' \setminus X$ is nonempty. Moreover, it is clear that $X'$ is a $(2^{-k}\e)$-separated set. To show that $X'$ is maximal, assume for a contradiction that there exists $v \in V$ such that $|x-v| \geq 2^{-k}\e$ for all $x\in X'$. Then, there exists $m\in\{1,\dots,n-l\}$ such that $x = c_{n-l-m+1}$. But then, by (\ref{eq1}), $x$ would end up in $X'$ which is false. The lower bound of (\ref{eq2}) follows from the $2^{-k}\e$-separability of $X'$ and the fact that $X'\setminus X \neq \emptyset$. For the upper bound of (\ref{eq2}) assume for a contradiction that there exists $x'\in X'$ such that $|x' -x| \geq 2^{1-k}\e$ for all $x\in X$. Then, clearly $x' \in V\setminus X \subset U$ and it follows that $d\geq 2^{1-k}\e$. But that contradicts the minimality of $k$ since $k-1$ could also work.
\end{proof}

\section{Flatness modules and flat pairs}\label{sec:flat}

In this section we examine a notion of local flatness introduced by Jones \cite{Jones-TST} and further developed by Schul \cite{Schul-Hilbert} and Badger and Schul \cite{BS3}. For the rest of this section we fix $C_0=300$, a finite set $V\subset \R^N$ consisting of $n$ elements, an $\e$-net $X\subset V$, and a $(2^{-k}\e)$-net $X'\subset V$ such that $X\subset X'$ and 
\[ 2^{-k}\e \leq \dist_H(X,X') <  2^{1-k}\e.\]

In \cite{Schul-Hilbert} Schul defined for each $v\in X$ a number that measures the (normalized) thickness of the thinnest cylinder that contains the set 
\[ B_{v,X'}:= B(v,C_0 2^{-k}\e) \cap X'.\] 
In our case, the numbers that we define measure the same thing with one difference; we will only consider a \emph{finite} set of cylinders that contain $B_{v,X'}$ (as opposed to \emph{all} cylinders) and we will choose the thinnest one. Nevertheless, the two flatness modules are close to each other; see Lemma \ref{lem:flat-comp}.

\subsection{Flatness modules}
For the rest of this section, $L$ is an integer with $L\geq 40C_0\sqrt{N-1}$ and $\mathcal{G} = \{-2L,\dots,2L\}^{N}$. Given two distinct points $y_1,y_2 \in \R^N$, we denote by $\ell(y_1,y_2)$ the unique line that passes through $y_1,y_2$ and define the normalized distance
\[ D(y_1,y_2) = \max_{z\in B_{v,X'}} 2^{k}\e^{-1} \dist(z, \ell(y_1,y_2)).\]
Given $v\in X$, define 
\[ \a_{v,X'} = \min_{x\in \phi(\mathcal{G})}\min_{y\in \phi(\mathcal{G})\setminus\{x\}} D(x,y)\]
where
$\phi(z) = C_02^{-k}\e L^{-1}z +v$.
Essentially, $\a_{v,X'} $ measures the smallest distance of the set $B_{v,X'}$ to all lines going through a fine grid of the cube $v+C_02^{1-k}\e[-1,1]^N$ normalized by $2^{-k}\e$. We set $\ell_{v,X'}$ to be the line passing through points $x,y \in \phi(\mathcal{G})$ with $x\neq y$ that minimize the quantity $D(x,y)$.

\begin{lem}\label{lem:alpha}
Number $\a_{v,X'}$ and line $\ell_{v,X'}$ can be computed in $O(\card X')$ time.
\end{lem}

\begin{proof}
We start by constructing the sets $B_{v,X'}$. We check over all points $u\in X'$ whether or not $|v-u| < C_0 2^{-k}\e$. Then the total time needed to construct the set $B_{v,X'}$ is $O(\card X')$. Since the set $\mathcal{G}$ is finite and independent of $n$, the calculation of $\a_{v,X'}$ and of $\ell_{v,X'}$ takes $O(1)$ time. This completes the proof of the lemma.
\end{proof}

\begin{rem}\label{rem:sharp}
Since for each $v\in X$ we have to go through the entire set $X'$ in order to construct $B_{v,X'}$, the computations of the preceding lemma are sharp; that is, there exists universal $C>0$ such that it takes at least $C\card{X'}$ time to compute $B_{v,X'}$.
\end{rem}

\begin{lem}\label{lem:flat-comp}
If there exists a line $\ell$ such that $\max_{x\in B_{v,X'}} \dist(x,\ell) \leq \frac1{20}2^{-k}\e$, then $\a_{v,X'} \leq 1/16$.
\end{lem}

\begin{proof}
Suppose that $\ell$ is a line with $\max_{x\in B_{v,X'}} \dist(x,\ell) \leq \frac1{20}2^{-k}\e$. Without loss of generality, we may assume that $\ell$ intersects the cube $v+ C_02^{-k}\e[-1, 1]^N$; otherwise we replace $\ell$ by a line parallel to $\ell$ that intersects the boundary of the cube and is closer to the set $B_{v,X'}$. Then, $\ell$ intersects the boundary of the cube $v+ C_02^{1-k}\e[-1, 1]^N$ on two points $p_1,p_2$ which have distance at least $2^{3/2}C_02^{-k}$. Now, we can find $p_1',p_2' \in \phi(\mathcal{G})$ such that for $i=1,2$
\[ |p_i-p_i'| \leq \tfrac12\sqrt{N-1}C_02^{-k}\e L^{-1} \leq \tfrac1{80}2^{-k}\e.\]
By the choice of $L$, we have that $p_1'\neq p_2'$.

Fix $x\in B_{v,X'}$. Then, by the choice of $L$,
\[ \dist(x,\ell(p_1',p_2')) \leq \dist(x,\ell) + \max_{i=1,2}|p_i-p_i'|  \leq \tfrac1{16}2^{-k}\e.\]
Therefore, $\a_{v,X'} \leq 2^{k}\e^{-1}\max_{x\in B_{v,X'}}\dist(x,\ell(p_1',p_2')) \leq 1/16$.
\end{proof}

\subsection{Flat Pairs}\label{sec:flatpairs}

Following \cite[Definition 2.4]{BNV}, for each $v\in X$, we define the set of \emph{flat pairs} $\Flat(v,X,X')$ to be the set of $\{v,v'\} \subset X$ such that
\begin{enumerate}
\item $\e \leq |v-v'| < C_02^{-k-1}\e$,
\item $\a_{v,X'}< 1/16$ and there exists an orientation of the line $\ell_{v,X'}$ such that $v'$ is the first point of $X \cap B(v, C_0 2^{-k-1}\e)$ to the left or to the right of $v$ with respect to the given orientation.
\end{enumerate}

Note that if $k$ is large, then the first condition fails immediately and no flat pairs exist. Here and for the rest of the paper, given a point $x \in \R^N$, we denote the first coordinate of $x$ by $\langle x\rangle_1$.

\begin{lem}\label{lem:flat}
The set $\Flat(v,X,X')$ can be computed in $O(\card X')$ time.
\end{lem}

\begin{proof}
We first compute the set
\[B'_{v,X}:=X \cap B(v, C_02^{-k-1}\e). \]
and as in Lemma \ref{lem:alpha}, it is easy to see that $B'_{v,X}$ can be computed in $O(\card X')$ time. Next we compute $\a_{v,X'}$ which also takes $O(\card X')$ time.

If $B'_{v,X}$ contains only $v$, or if $\a_{v,X'} > 1/16$ then we set $\Flat(v,X,X') = \emptyset$. 

Assume now that $B'_{v,X}$ contains more than $v$ and that $\a_{v,X'} \leq 1/16$. Having the line $\ell_{v,X'}$ from the previous subsection, we apply an isometric affine transformation 
\[ \Psi : \R^N \to \R^N,\qquad \Psi(x) = A(x-v)\]
such that $A$ is an orthogonal matrix, $\Psi(v)=0$, and $\Psi$ maps the line generated by $\ell_{v,X'}$ onto the line passing trough the origin that is parallel to the $x_1$-axis. The matrix $A$ can be generated in $O(1)$ time. We compute the set $\{\langle \Psi(v')\rangle_1 : v' \in B'_{v,X}\}$ and enumerate them from lowest to highest. By \cite[Lemma 8.3]{BS3} no two first coordinates can be equal. 
\begin{itemize}
\item If there exists $v' \in B'_{v,X}$ with $\langle \Psi(v')\rangle_1<0$, then choose the point $v' \in B'_{v,X'}$ with the highest first coordinate that is negative, and place $\{v,v'\}$ in $\Flat(v,X,X')$.
\item If there exists $v' \in B'_{v,X}$ with $\langle \Psi(v')\rangle_1>0$, then choose the point $v' \in B'_{v,X}$ with the lowest first coordinate that is positive, and place $\{v,v'\}$ in $\Flat(v,X,X')$.
\end{itemize}
Since $\a_{v,X'} \leq 1/16$ and since $B'_{v,X'}\setminus \{v\}$ is nonempty, at least one of the two cases above must hold. If both hold, then we perform both actions. 

Since $X$ is $\e$-separated, the cardinality of $B'_{v,X}$ depends only on $N$, and the above computations can be done in $O(1)$ time. Hence, the computation of $\Flat(v,X,X')$ can be done in $O(\card X')$ time and the proof is complete.
\end{proof}

\section{ATSP algorithm and proof of Theorem \ref{thm:main}}\label{sec:graphs}

Before starting the algorithm, we calculate $R_0 = 5\max\{|v| : v\in V\}$. This takes $O(n)$ time. Here and for the rest of this section, all constants depend only on the dimension $N$.

\subsection{Step 1}
We set $n_1=1$, $V_1=\{v_1\}$, and $U_1 = V\setminus \{v_1\}$. These definitions take $O(n)$ time. If $U_1=\emptyset$, then we move to the Final Step; see Section \ref{sec:finalstep}. Otherwise, we apply the algorithm of Lemma \ref{lem:nets} for $X=V_1$ and $\e=2^{-1}R_0$ and we obtain in $O(n^2)$ time an integer $n_2>1$ and a $(2^{-n_2}R_0)$-net $V_2$ along with a set $U_2 = V\setminus V_2$. Now, we apply the algorithm of Lemma \ref{lem:alpha} to obtain in $O(n)$ time a number $\alpha_{v_1,V_2}$. If $\alpha_{v_1,V_2}> 1/16$, then we set $\mathcal{N}_1 = \{v_1\}$ and $\mathcal{F}_1 = \emptyset$. Otherwise, we set $\mathcal{F}_1 = \{v_1\}$ and $\mathcal{N}_1 = \emptyset$. We also set $E_1 = \emptyset$ and $G_1 = (V_1,E_1)$. Step 1 takes $O(n^2)$ time.

\subsection{Step $k+1$}

Suppose that for some $k\in\N$ we have produced
\begin{itemize}
\item an integer $n_k\geq 1$, 
\item a $(2^{-n_k}R_0)$-net $V_k$ and a nonempty set $U_k = V\setminus V_k$, 
\item two disjoint sets $\mathcal{N}_k,\mathcal{F}_k \subset V_k$ with $\mathcal{N}_k\cup \mathcal{F}_k = V_k$, 
\item a connected graph $G_k = (V_k, E_k)$ such that $\card{E}_k \leq 2 \card{V_k}$,
\item an integer $n_{k+1}> n_k$, and 
\item a $(2^{-n_{k+1}}R_0)$-net $V_{k+1}$ and a set $U_{k+1} = V\setminus V_{k+1}$.
\end{itemize}

If $U_{k+1} = \emptyset$, then this paragraph can be skipped. Otherwise, we apply the algorithm of Lemma \ref{lem:nets} with $X=V_{k+1}$ and $\e=2^{-n_{k+1}}R_0$ and we obtain in $O(n^2)$ time an integer $n_{k+2}>n_{k+1}$, a $(2^{-n_{k+2}}R_0)$-net that contains $V_{k+1}$, and a set $U_{k+2} = V\setminus V_{k+2}$.

Next, we produce three new families $\mathcal{F}_{k+1},\mathcal{N}_{k+1}, E_{k+1}$.

\subsubsection{Families $\mathcal{N}_{k+1}$ and $\mathcal{F}_{k+1}$} We perform this part of Step $k+1$ only if $U_{k+1}\neq \emptyset$, otherwise we omit it. For each $v \in V_{k+1}$ we compute $\a_{v,V_{k+2}}$ and 
\begin{itemize}
\item if $\a_{v,V_{k+2}} > 1/16$, then we place $v$ in $\mathcal{N}_{k+1}$;
\item if $\a_{v,V_{k+2}} \leq 1/16$, 
then we place $v$ in $\mathcal{F}_{k+1}$.
\end{itemize}
By Lemma \ref{lem:alpha}, the constructions of sets $\mathcal{N}_{k+1},\mathcal{F}_{k+1}$, takes a total of $O(n^2)$ time.

\bigskip

It remains to construct $E_{k+1}$. The new edges will come from three sources; from old edges in $E_k$, around points in $\mathcal{F}_k$, and around points in $\mathcal{N}_k$.

\subsubsection{Edges coming from $E_k$} Recall from Section \ref{sec:flat} that $C_0=300$. For each element $e = \{u,u'\} \in E_k$, we compute $|u-u'|$ and the two numbers $\a_{u,V_{k+1}}$ and $\a_{u',V_{k+1}}$. 

If $|u-u'| \geq C_0 2^{-n_{k+1}-1}R_0$, or if both $\a_{u,V_{k+1}}$ and $\a_{u',V_{k+1}}$ are more than $1/16$, then we place $e$ in $E_{k+1}^{(1)}$. This case takes $O(n)$ time. We set $V_{k+1}(e) = \emptyset$; this set will not play a role in the algorithm but we define it for consistency.

If $|u-u'| < C_0 2^{-n_{k+1}-1}R_0$ and at least one of $\a_{u,V_{k+1}}$, $\a_{u',V_{k+1}}$ (say $\a_{u,V_{k+1}}$) is less or equal to $1/16$, then we calculate $B_{u,V_{k+1}}\cup B_{u',V_{k+1}}$ and $\ell_{u,V_{k+1}}$. We also calculate an isometric affine map $\Psi$ that takes $u$ to the origin, takes $u'$ to a point with positive first coordinate, and takes the line $\ell_{u,V_{k+1}}$ onto the $x_1$-axis. Then we determine those points $u_1,\dots,u_l \in B_{u,V_{k+1}}\cup B_{u',V_{k+1}}$ such that the first coordinates satisfy
\[ 0 = \langle\Psi(u)\rangle_1 < \langle\Psi(u_1)\rangle_1 < \cdots < \langle\Psi(u_l)\rangle_1 < \langle\Psi(u')\rangle_1.\]
By \cite[Lemma 8.3]{BS3} we know that no distinct points $\Psi(u_i), \Psi(u_j)$ have equal first coordinates so the inequalities above are justified. We set $V_{k+1}(e) = \{u_1,\dots,u_l\}$ and we add the edges $\{u, u_1\}, \dots, \{u_l,u'\}$ in $E_{k+1}^{(1)}$. The calculation of $\a_{u,V_{k+1}}$ and $\a_{u',V_{k+1}}$, as well as the calculation of $B_{u,V_{k+1}}$, $B_{u',V_{k+1}}$ takes $O(n)$ time. The calculation of $\ell_{u,V_{k+1}}$, $\Psi$, and the subsequent ordering takes $O(1)$. (However, $l$ is bounded by a constant independent of $N$; see \cite[Remark 3.2]{BNV}). Therefore, this case takes $O(n)$ time.

In total, the computation of $E_{k+1}^{(1)}$ takes $O(n)\card{E_{k}} = O(n^2)$ time, with the associated constants depending only on $N$. Moreover,
\begin{equation}\label{eq:1}
\card{E^{(1)}_{k+1}} \leq \card{E_k} + 2\sum_{e \in E_k} \card{V_{k+1}(e)}.
\end{equation}

\begin{rem}\label{rem:est1}
Let $e,e' \in E_k$ satisfy the assumptions of the second case. Then by \cite[Lemma 3.3]{BNV}, $V_{k+1}(e)\cap V_{k+1}(e') = \emptyset$. 
\end{rem}

\subsubsection{Edges coming from points in $\mathcal{F}_k$}
For each element $u$ in $\mathcal{F}_k$ we calculate the closed ball $B_{u,V_{k+1}}$ and the associated line $\ell_{u,V_{k+1}}$ which takes $O(n)$ time. Calculate an isometric affine map $\Psi:\R^N \to \R^N$ that takes the line $\ell_{u,V_{k+1}}$ to the $x_1$-axis, and $u$ to a point with zero first coordinate. This takes $O(1)$ of time. If there exists $u',u'' \in B_{u,V_{k+1}}\cap V_k$ such that $\langle\Psi(u')\rangle_1 < 0 < \langle\Psi(u'')\rangle_1$, then no new edges are obtained from $u$ and we move to the next element of $\mathcal{F}_k$.  Otherwise, we work as follows.
\begin{enumerate}
\item If there exists no $u' \in B_{u,V_{k+1}}\cap V_k$ such that $\langle\Psi(u')\rangle_1 < 0 $, then enumerate the points $u_1,\dots,u_l \in \overline{B}(u, 2^{1-n_{k+1}}R_0)$ such that 
\[ \langle\Psi(u_1)\rangle_1 < \cdots < \langle\Psi(u_l)\rangle_1 < 0 = \langle\Psi(u)\rangle_1\]
and add the edges $\{u_1,u_2\}, \cdots, \{u_l,u\}$ in $E_{k+1}^{(2)}$.
\item If there exists no $u' \in B_{u,V_{k+1}}\cap V_k$ such that $\langle\Psi(u')\rangle_1 > 0 $, then enumerate the points $u_1,\dots,u_l \in \overline{B}(u, 2^{1-n_{k+1}}R_0)$ such that 
\[ \langle\Psi(u_1)\rangle_1 > \cdots > \langle\Psi(u_l)\rangle_1 > 0 = \langle\Psi(u)\rangle_1\]
and add the edges $\{u_1,u_2\}, \cdots, \{u_l,u\}$ in $E_{k+1}^{(2)}$.
\end{enumerate}
If none of the above is true, then no new edges are obtained from $u$ and we move to the next element of $\mathcal{F}_k$. If both are true, we perform both of them. By \cite[Remarks 3.4 and 3.5]{BNV}, we have that $l \leq 6$. Therefore, the computation of points $u_i$, their enumeration, and addition of edges takes $O(1)$ time. The set of all the points in $V_{k+1}\setminus V_k$ coming from both of these cases above (if any) is denoted by $V_{k+1}(u)$.

We repeat the same process for all points in $\mathcal{F}_k$ and the construction of $E_{k+1}^{(2)}$ is done in $(\card{\mathcal{F}_k})O(n) = O(n^2)$ time. Moreover,
\begin{equation}\label{eq:2}
\card{E^{(2)}_{k+1}} \leq 2\sum_{u \in \mathcal{F}_k} \card{V_{k+1}(u)}.
\end{equation}

\begin{rem}\label{rem:est2}
By \cite[Lemma 3.6]{BNV}, if $u,u' \in \mathcal{F}_k$, then $V_{k+1}(u)\cap V_{k+1}(u') = \emptyset$. Moreover, if $e \in E_k$ and $u \in \mathcal{F}_k$, then $V_{k+1}(u)\cap V_{k+1}(e) = \emptyset$.
\end{rem}

Before proceeding to the final set $E_{k+1}^{(3)}$ we enumerate
\[ \mathcal{N}_{k} = \{u_1,\dots,u_l\}.\]

\subsubsection{Edges coming from $\mathcal{N}_k$: first step}\label{sec:3-1}
Define the set 
\[ V_{k+1}'(u_1) = B(u_1, C_02^{-n_{k+1}}R_0)\cap V_{k+1}.\] 
The construction of this set requires $O(n)$ time and its cardinality is $O(1)$. We now go over the set  $E_{k+1}^{(1)}\cup E_{k+1}^{(2)}$ formed above and if a point $v \in V_{k+1}'(u_1)$ is contained in an edge of $E_{k+1}^{(1)}\cup E_{k+1}^{(2)}$, then we remove it from the set $V_{k+1}'(u_1)$ and in this way we form the set $V_{k+1}(u_1)$ in $O(n)$ time. By the doubling property of $\R^N$, $\card{V_{k+1}(u)} = O(1)$. Moreover, from the previous subsections we already have that every point in $V_k$ is already contained in an edge of $E_{k+1}^{(1)}\cup E_{k+1}^{(2)}$. Therefore,
\[ V_{k+1}(u_1) \subset V_{k+1} \setminus \left( V_k \cup \bigcup_{e\in E_k}V_{k+1}(e) \cup \bigcup_{u\in \mathcal{F}_k}V_{k+1}(u) \right).\]
If the set $V_{k+1}(u_1)$ is empty, then we are done with $u_1$ and we move to $u_2$. Otherwise, we construct the set
\[ \Flat(k+1,u_1) = \bigcup_{u\in V_{k+1}(u_1)}  \Flat(u,V_k,V_{k+1}).\]
Since there are $O(1)$ points in $V_{k+1}(u)$, by Lemma \ref{lem:flat} it takes a total of $O(n)$ time to determine $\bigcup_{u\in V_{k+1}(u_1)}\Flat(u,V_k,V_{k+1})$. We make sure that the latter set does not contain the same edge twice and we denote it by $\Flat(k+1,u_1)$. Since the cardinality of $\Flat(k+1,u_1)$ is $O(1)$, this can be done in $O(1)$ time. Define $\tilde{V}_{k+1}(u_1)$ to be the union of $V_{k+1}'(u_1)$ and the set of all points in $V_{k+1}$ that belong to an edge in $\Flat(k+1,u_1)$. 

Next, we form the graph $G_{k+1}(u_1) = (\tilde{V}_{k+1}(u_1), \Flat(k+1,u_1))$. We apply the algorithm of Lemma \ref{lem:graph1} which in $O(1)$ time returns the components of $G_{k+1}(u_1)$. If there is only one component, then we set $E_{k+1}(u_1) :=  \Flat(k+1,u_1)$ (making sure in $O(1)$ time that no edge appears twice) and we add all edges of $E_{k+1}(u_1)$ (if any) in $E_{k+1}^{(3)}$ and we move to $u_2$. If there are $p$ components with $p\geq 2$, then we fix a point $u_{1,i}$ in each component and we set
\[ E_{k+1}(u_1) := \{\{u_{1,1},u_{1,2}\}, \dots, \{u_{1,1}, u_{1,p}\}\} \cup \Flat(k+1,u_1)\]
and we add all edges of $E_{k+1}(u_1)$ in $E_{k+1}^{(3)}$. 

Recall that each point in $V_{k+1}(u_1)$ can be in at most 2 edges of $\Flat(k+1,u_1)$. Therefore, it is not hard to see that
\[ \card(E_{k+1}(u_1)) \leq 2\card(V_{k+1}(u_1)).\]

If $l=1$, then the definition of $E_{k+1}^{(3)}$ is complete we move to Section \ref{sec:wrap-up}; otherwise and we move to $u_2$.

\subsubsection{Edges coming from $\mathcal{N}_k$: inductive step} Suppose that we have already added edges to $E_{k+1}^{(3)}$ from vertices $u_1,\dots,u_{p-1}$ with $p\leq l$. The definition of the set $V_{k+1}(u_p)$ is the same as with $V_{k+1}(u_1)$ with one important exception: this time we go over the set $E_{k+1}^{(1)}\cup E_{k+1}^{(2)} \cup E_{k+1}^{(3)}$ instead of just $E_{k+1}^{(1)}\cup E_{k+1}^{(2)}$. Therefore,
\begin{equation}\label{eq:est3}
V_{k+1}(u_1) \subset V_{k+1} \setminus \left( V_k \cup \bigcup_{e\in E_k}V_{k+1}(e) \cup \bigcup_{u\in \mathcal{F}_k}V_{k+1}(u) \cup \bigcup_{i=1}^{p-1}V_{k+1}(u_i) \right).
\end{equation}
The rest of the construction is verbatim the same as before and we add to $E_{k+1}^{(3)}$ a new (possibly empty) set of edges $E_{k+1}(u_p)$ with
\begin{equation}\label{eq:3}
\card(E_{k+1}(u_p)) \leq 2\card(V_{k+1}(u_p)).
\end{equation}

\subsubsection{Definition of $E_{k+1}$}\label{sec:wrap-up}
Going over the entire set of $\mathcal{N}_k$, we finally define $E^{(3)}_{k+1}$ in $O(n^2)$ time. From Remark \ref{rem:est1}, Remark \ref{rem:est2}, and (\ref{eq:est3}), we have that the three sets $E^{(1)}_{k+1}$, $E^{(2)}_{k+1}$, $E^{(3)}_{k+1}$ are mutually disjoint and by estimates (\ref{eq:1}), (\ref{eq:2}), (\ref{eq:3}),
\[ \card(E^{(1)}_{k+1} \cup E^{(2)}_{k+1} \cup E^{(3)}_{k+1}) \leq \card{V_k} + 2 \card(V_{k+1}\setminus V_k) \leq 2 \card{V_{k+1}}.\]
We now set $E_{k+1} = E^{(1)}_{k+1} \cup E^{(2)}_{k+1} \cup E^{(3)}_{k+1}$ (the enumeration takes $O(n)$ time) and the inductive step is complete.

\subsection{Final Step}\label{sec:finalstep} To finish the proof of Theorem \ref{thm:main}, we remark that for the final graph $G_m = (V_m,E_m)$ that we obtained, we have that $V_m=V$ and $\card{E_m} = O(n)$. Therefore, the algorithm of Lemma \ref{lem:tour}, produces the desired tour in $O(n^3)$ time. This completes the proof of Theorem \ref{thm:main}.

\subsection{Sharpness of exponent 3}\label{sec:sharp}
Here we discuss why the time complexity $O(n^3)$ is sharp. We will do this by analyzing the computations of the flat pairs and flatness modules. Note that the complexity of $O(n^3)$ appears during many other steps of the algorithm as well (e.g. Lemma \ref{lem:graph1} and Lemma \ref{lem:tour}). 

Let $n$ be a positive integer and suppose that for all $i \in\{1,\dots,n\}$ , $\card V_i = i$. At step $k$, we perform the calculation of sets $\Flat(v,V_k,V_{k+1})$ for each $v\in V_k$ which, by Remark \ref{rem:sharp}, requires $C\card{V_{k+1}}$ time for each such $v$, and for some fixed constant $C>0$. Therefore, the computation of all flat pairs at step $k$ requires $C(\card V_k)(\card V_{k+1})$ computations. Therefore, our total computations over all steps are
%\begin{align*}
	\[ C\sum_{k=1}^{n-1}  (\card V_k) (\card V_{k+1}) = C\sum_{k=1}^{n-1} k(k+1) \geq 
	%C\sum_{i=n/2}^{n-1} i(i+1)  \geq 
	C\left( \frac{n}{2} \right)^2 \frac{n}{2} = C\frac{n^3}{8}.\]
%\end{align}
Since the calculation of the flatness modules for each net $V_i$ is but a part of the algorithm, the total computation time is at least a fixed multiple of $n^3$.

%Let $n$ be a large multiple of 4 and let $V = \{1, 2^{-1},\dots, 2^{2-n}, 0\}$. Suppose that $V_1=\{0\}$. We need $n-1$ computations to obtain $n_2=1$ and $V_2=\{0,1\}$, we need $2(n-2)$ computations to obtain $n_3=2$ and $V_3 = \{0,1/2,1\}$, etc. In total, the number of computations we perform to calculate the nets $V_1\subset\cdots\subset V_n$ is 
%\begin{align*} 
%\sum_{i=1}^{n-1} i(n-i) \geq \sum_{i=n/4}^{n/2} i(n-i) \geq \left(\frac{n}4\right)^2\frac{n}{2} = %\frac{n^3}{32}.
%\end{align*}
%Since the calculation of nets $V_i$ is but a part of the algorithm, the total computation time is at least a fixed multiple of $n^3$.

\bibliography{tsp-refs}
\bibliographystyle{amsbeta}

\end{document}